\documentclass[11pt,letterpaper]{article}
\usepackage[subtle, mathspacing=normal]{savetrees}
\usepackage[utf8]{inputenc}
\usepackage{times}
\usepackage[ruled,vlined,commentsnumbered,titlenotnumbered]{algorithm2e}
\usepackage{url}
\usepackage{fullpage}
\usepackage{pslatex} 
\usepackage{mathdots} 
\usepackage{amsmath}
\usepackage{algorithmic}
\usepackage{amssymb}

\usepackage{amsthm}
\usepackage{color}
\usepackage{array}
\usepackage{xy}
\usepackage{setspace}
\usepackage{multicol}
\usepackage{hyperref}
\usepackage{cite}
\usepackage[margin=1in,letterpaper]{geometry}
\usepackage{thm-restate}

\usepackage{accents}

\newcommand{\defn}[1]{\emph{\textbf{{#1}}}}

\usepackage{todonotes}

\newcommand{\E}{\mathbb{E}}

\renewcommand{\paragraph}[1]{\vspace{.5 cm} \noindent \textbf{#1} }

\usepackage{enumitem}
\usepackage{thmtools}
\usepackage{thm-restate}

\newtheoremstyle{slanted}
{3pt}
{3pt}
{\slshape}
{}
{\bfseries}
{.}
{.5em}
{}
\theoremstyle{slanted}
\newtheorem{theorem}{Theorem}
\newtheorem{lemma}[theorem]{Lemma}
\newtheorem{claim}[theorem]{Claim}
\newtheorem{proposition}[theorem]{Proposition}

\newtheorem{corollary}[theorem]{Corollary}

\begin{document}

%
\title{A Simple and Combinatorial Approach to Proving  Chernoff Bounds and Their Generalizations \\ (with Almost no Algebra)}
\author{William Kuszmaul\footnote{CMU. \texttt{kuszmaul@cmu.edu}}}


\date{}

\maketitle

\begin{abstract}
The Chernoff bound is one of the most widely used tools in theoretical computer science. It's rare to find a randomized algorithm that \emph{doesn't} employ a Chernoff bound in its analysis.

The standard proofs of Chernoff bounds are beautiful but in some ways not very intuitive. In this paper, I'll show you a different proof that has four features:
\begin{itemize}[noitemsep]
\item the proof offers a strong intuition for why Chernoff bounds look the way that they do;
\item the proof is user-friendly and (almost) algebra-free;
\item the proof comes with matching lower bounds, up to constant factors in the exponent;
\item the proof extends to establish generalizations of Chernoff bounds in other settings.
\end{itemize}

The ultimate goal is that, once you know this proof (and with a bit of practice), you should be able to confidently reason about Chernoff-style bounds in your head, extending them to other settings, and convincing yourself that the bounds you're obtaining are tight (up to constant factors in the exponent).

\end{abstract}
\vfill
\pagebreak

\section{Introduction}

The Chernoff bound is one of the most widely used tools in theoretical computer science. It's rare to find a randomized algorithm that \emph{doesn't} employ a Chernoff bound in its analysis.

The standard proofs of Chernoff bounds are beautiful but in some ways not very intuitive. In this paper, I'll show you a different proof that, as far as I can tell, has not appeared in the literature. The proof will have four noteworthy features:
\begin{enumerate}
\item \textbf{Intuition: }The proof offers a clean intuition for why Chernoff bounds have the shapes that they do.
\item \textbf{Lower Bounds: }The proof comes with matching \emph{lower bounds}. In fact, one way to perform the proof is to first prove lower bounds, and then directly argue that those lower bounds are tight (up to constants in the exponent).
\item \textbf{User-Friendliness: }Most proofs of Chernoff bounds require various identities (and Taylor series approximations) to obtain the final user-friendly bound. The proof here will lead directly to a user-friendly bound without requiring intermediate algebra.
\item \textbf{Generality: }The proof can be extended to establish many of the classical generalizations of Chernoff bounds (e.g., Hoeffding bounds, Azuma's inequality, Bernstein-type inequalities, Bennett's inequality, etc.). In other words, the proof isn't just a party trick. 
\end{enumerate}

The proof will also have a noteworthy drawback: It will give the correct bound \emph{up to constant factors in the exponent}. If those constant factors are important to you, then you'll want to use other derivations.

\paragraph{How to think about Chernoff bounds.}
Before getting into proofs, it is worth first reviewing the basic bounds that we will wish to prove. Let $X_1, X_2, \ldots, X_n$ be independent $0$-$1$ random variables satisfying $\Pr[X_i = 1] = p$ for some $p \le 1/2$. Let $X = \sum_{i = 1}^n X_i$ and let $\mu = \E[X] = np$. 

The Chernoff bound tells us that the probability of $X$ deviating substantially above its mean $\mu$ is small. That is, we get an upper bound on $\Pr[X \ge \mu + t]$ as a function of $t$. 

Chernoff bounds are presented in many different forms, and students often have trouble figuring out which version to memorize (the result is that many students end up using Wikipedia as a regular reference). So what is the right way to think about Chernoff bounds?

The first thing to know is that $n$ and $p$ are red herrings. The only parameter that actually matters is $\mu$. In fact, one can obtain tight Chernoff bounds in all regimes by just remembering two simple bounds. The first is the \defn{small-deviation bound}, which says that
\begin{equation}
\Pr[X \ge \mu + k\sqrt{\mu}] = \frac{1}{2^{\Theta(k^2)}}
\label{eq:small}
\end{equation}
for any $k = O(\sqrt{\mu})$ satisfying $k \sqrt{\mu} \le n$. The second is the \defn{large-deviation bound}, which says that
\begin{equation}
\Pr[X \ge \mu + r\mu] = \frac{1}{\Theta(r)^{\Theta(r \mu)}}
\label{eq:large}
\end{equation}
for any $r \ge 1$ satisfying $1 \le \mu + r\mu \le n$. 

There are three things to notice about these bounds. First, as I mentioned earlier, we are not worrying about what the constants are in the exponents. As theoreticians, we are almost always interested in asymptotic deviations (i.e., $\Pr[X \ge \mu + \Omega(t)]$ for various $t$), so the constant in the exponent typically doesn't matter. Second, the fact that the exponents are in $\Theta$-notation, rather than $\Omega$-notation, is no coincidence: both bounds turn out to be tight up to constant factors in the exponent. Third, the two bounds become equivalent when we consider $\Pr[X \ge \mu + \Theta(\mu)]$, so there is a smooth transition from one regime to the other.

It is also worth taking a few moments to internalize the shapes of these bounds. The small-deviation bound, should be viewed as telling us something about standard deviations. It turns out that the standard deviation of $X$ is guaranteed to be $\Theta(\sqrt{\mu})$, regardless of $n$ and $p$ (as long as $p \le 1/2$). Thus the bound says that the probability of being $k$ standard deviations above the mean shrinks at a rate of $1 / 2^{\Theta(k^2)}$. In fact, we will later see that the previous sentence continues to be true in much more general settings, and that this is the source of what is known as Bennett's inequality (we will come back to this later).

The large-deviation bound also takes an interesting shape. It is much stronger than most students would guess it should be. A priori, students typically assume that the bound should be something like $1 / 2^{\Theta(r \mu)}$. This is correct when $r = \Theta(1)$, but when $r$ is larger, \eqref{eq:large} gets stronger, replacing the denominator of $2$ with $r$. 

Although the above Chernoff bounds are stated in the case where $X_1, X_2, \ldots, X_n$ are identically distributed 0-1 random variables, the same upper bounds hold for any independent real-valued random variables $X_1, X_2, \ldots, X_n \in [0, 1]$ satisfying $\E[\sum_i X_i] = \mu$. The corresponding lower bounds do not necessarily hold in this more general setting (for example, it might be that $X_1, X_2, \ldots, X_n$ are all deterministically $1$, so $\Pr[X = \mu] = 1$), but we shall see that it is a relatively simple task to reason about when the lower bounds do or do not hold.

For readers that wish to apply \eqref{eq:small} and \eqref{eq:large} with explicit constants, we remark that a careful checking of the constants yields the following bounds. For any independent $X_1, X_2, \ldots, X_n \in [0, 1]$ satisfying $\E[\sum_i X_i] = \mu$, we have $$\Pr[X \ge \mu + k \sqrt{\mu}] \le 2^{-k^2/2}$$ for all $1 \le k \le \sqrt{\mu}$, and $$\Pr[X \ge r\mu] \le r^{-r\mu /4}$$ for all $r \ge 2$. 

\paragraph{Paper outline.} In the body of the paper, we will present the new Chernoff bound derivation from four different perspectives: 

\begin{itemize}
    \item \textbf{The One-Page Version (Section \ref{sec:onepage}). }We begin in Section \ref{sec:onepage} with a bare-bones version of the proof---a one-page self-contained analysis that focuses on the special case where we have $n$ fair coin flips. This version of the proof is designed for readers who like to read first and digest after. It does not concern itself with side-quests such as proving lower bounds or highlighting intuition. Additionally, to simplify the presentation, and because we are focusing only on \emph{fair} coin flips, we follow the convention in both this section and the next that each $X_i$ is in $\{-1, 1\}$ rather than $\{0, 1\}$.
    \item \textbf{The Extended Edition (Section \ref{sec:fair}). }In Section \ref{sec:fair}, we present the same proof again, but with additional commentary to motivate the steps and explain what's going on at a higher level. This version of the proof is designed for readers who like to digest as they read. It includes a focus on intuition, as well as a small side-quest to prove matching lower bounds. In fact, quite happily, the lower-bound proof serves as a strong motivator for why the \emph{upper-bound proof} should follow the structure that it does.
    \item \textbf{Bias Coin Flips and the Large-Deviation Regime (Section \ref{sec:biased}).} Section \ref{sec:biased} extends the proof to the setting of biased coin flips, where each coin has some probability $p \le 1/2$ of being heads. This allows us to present the large-deviation bound (Equation \ref{eq:large}). The proof follows a very similar structure to the small-deviation case, and comes once again with matching lower bounds.
    \item \textbf{Generalizing to an Adaptive Bennett's Inequality (Section \ref{sec:general}). } Having proven both the small and large deviation bounds for classical Chernoff bounds, we turn our attention in Section \ref{sec:general} to proving a powerful generalization of Chernoff bounds, namely, an adaptive version of Bennett's Inequality. Here, we are intentionally picking one of the most ``heavy-weight'' generalizations of Chernoff bounds. The point is to demonstrate how, with the same basic techniques that we used to proof the basic Chernoff bounds, and by just filling in a few more details, we can walk away with bounds that would traditionally be viewed as out of reach for combinatorial proofs.
\end{itemize}

We remark that the one-page proof in Section \ref{sec:onepage} is short but is not necessarily the right starting place for every reader. Some readers (especially students seeing Chernoff bounds for the first time) may wish to start with Sections \ref{sec:fair} and \ref{sec:biased}, and then to optionally add on additional sections from there. 

The ultimate goal of the paper is that, once you have digested the proof  (and with a bit of practice), you should be able
to confidently reason about Chernoff-style bounds in your head, extending them to other settings, and
convincing yourself that the bounds you’re obtaining are tight (up to constant factors in the exponent).

\paragraph{Historical context and past work.} 
Chernoff bounds first appeared in the literature in 1952 paper by Herman Chernoff \cite{chernoff1952measure} (although Chernoff himself attributes them to Herman Rubin \cite{Chernoff14}). The bounds and their generalizations have also been independently formulated by many other authors, including Kazuoki Azuma \cite{Azuma67}, Wassily Hoeffding \cite{Hoeffding94}, and Sergei Bernstein \cite{Bernstein1}. 

The classical proof of Chernoff bounds proceeds by applying Markov's inequality to the moment-generating function of a random variable. This is an important technique, that has also served as a core foundation for much of the work on concentration inequalities in statistics and probability theory \cite{Deviations1, Deviations2, Deviations3, Deviations4, Deviations5, Deviations6, Deviations7, Deviations8, Deviations9, Deviations10, Deviations11, Deviations12, Deviations13, Deviations14, McDiarmid89} (see \cite{DeviationsBook} or \cite{chung2006concentration} for a survey). There have also been several other proofs \cite{impagliazzo2010constructive, chvatal1979tail, steinke2017subgaussian, morin2017encoding}, using techniques from areas ranging from coding theory \cite{morin2017encoding} to differential privacy \cite{steinke2017subgaussian}; see Mulzer's survey \cite{mulzer2018five} for a description of the five main proof approaches that have been proposed. 

Most of these proof approaches \cite{mulzer2018five} struggle to generalize to more diverse settings---indeed, besides the classical moment-generating function approach, only one of the other approaches covered in \cite{mulzer2018five}, namely the proof of \cite{impagliazzo2010constructive}, appears to extend to prove Azuma's inequality (which, in turn, is weaker than the generalization that we prove in Section \ref{sec:general}). Additionally, all of the previous proofs \cite{mulzer2018five} share the unfortunate property that, in order to get to a user-friendly bound (i.e., to either of Equations \eqref{eq:small} or \eqref{eq:large}), one must first apply algebraic identities such as Taylor expansions. 

There is, not surprisingly, much less of a focus on lower bounds than there are on upper bounds. The classical moment-generating-function argument can be extended (non-trivially) to obtain essentially matching lower bounds, see, e.g., \cite{blog, wainwright2019high}. The simple combinatorial approach to proving lower bounds that is taken in the current paper does not appear to have been observed in past work, and the most general lower bound that we prove (Section \ref{sec:generallower}) does not appear to follow from the standard lower-bound techniques \cite{blog, wainwright2019high}.

\newpage

\section{Fair Coin Flips: The Bare-Bones Proof}\label{sec:onepage}

In this section, we consider a sum $X = \sum_{i = 1}^n X_i$ of independent unbiased coin flips $X_i \in \{1, -1\}$, and we prove that $\Pr\left[X \ge k \sqrt{n}\right] \le 2^{-\Omega(k^2)}$. Our starting point is a simple extension of Chebyshev's inequality:
\begin{lemma}[Extended Chebyshev]
For any $k \ge 1$, we have $\Pr[\max_j \sum_{i = 1}^j X_i \ge k \sqrt{n}] \le \frac{2}{k^2}$.
\label{lem:kolmogorov}
\end{lemma}
\begin{proof}
By Chebyshev's inequality, we have $\Pr\left[\sum_{i = 1}^n X_i \ge k \sqrt{n}\right] \le 1 / k^2$. Thus, it suffices to show that
\begin{equation}
\Pr\left[\max_j \sum_{i = 1}^j X_i \ge k \sqrt{n}\right] \le 2 \Pr\left[\sum_{i = 1}^n X_i \ge k \sqrt{n}\right].
\label{eq:chebyshev2factor}
\end{equation}
On the other hand, \eqref{eq:chebyshev2factor} follows from the following simple observation: If there exists $j \ge 0$ such that $\sum_{i = 1}^j X_i \ge k \sqrt{n}$, then with probability at least $0.5$ we have that $\sum_{i = j + 1}^n X_i \ge 0$, and thus that $\sum_{i = 1}^n X_i \ge k \sqrt{n}$.
\end{proof}
Using Lemma \ref{lem:kolmogorov}, we can derive a very simple (but already interesting) concentration bound:
\begin{lemma}[Poor Man's Chernoff Bound]
For $k \ge 1$,  $\Pr[X \geq k \sqrt{n}] = 2^{-\Omega(k)}.$
\label{lem:poorchernoff}
\end{lemma}
\begin{proof}
    For $j \ge 0$, let $t_j$ be the smallest index such that $\sum_{i = 1}^{t_j} X_i \ge j \cdot (2\sqrt{n} + 1)$, if such an index exists. For $j \ge 1$, if $t_j$ exists, then $\sum_{i = t_{j - 1} + 1}^{t_j} X_i \ge 2 \sqrt{n}$. So, if we condition on $t_{j - 1}$ existing, we can apply Lemma \ref{lem:kolmogorov} to $X_{t_{j - 1} + 1}, \ldots, X_n$ to get $\Pr[t_{j} \text{ exists} \mid t_{j - 1} \text{ exists}] \le \frac{1}{2}$. By induction on $j$, this implies $\Pr[t_j \text{ exists}] \le 2^{-\Omega(j)}$.
\end{proof}
In addition to the result above, we will need a Chernoff bound for sums of geometric random variables.
\begin{lemma}[Chernoff Bound for Geometric R.V.s]
Let $Y_1, Y_2, \ldots, Y_n$ be independent real-valued random variables and let $p \in (0, 1)$. If each $Y_i$ satisfies $\Pr[Y_i \ge j] \le p^j$ for all $j \in \mathbb{N}$, then $\Pr[\sum_i Y_i \ge 2n] \le (4p)^{n}.$
\label{lem:geo}
\end{lemma}
\begin{proof}
If $\sum Y_i \ge 2n$ then $\sum \lfloor Y_i \rfloor \ge n$. Thus there must exist $\vec{a} = (a_1, a_2, \ldots, a_n) \in (\mathbb{N} \cup \{0\})^n$ such that $\sum_i a_i = n$ and such that $\max(Y_i, 0) \ge a_i$ for each $i \in [n]$. Let $A$ denote the set of possible vectors $\vec{a}$. For a given $\vec{a} \in A$,
$$\Pr[\max(Y_i, 0) \ge a_i \text{ for all } i] \le \prod_i \Pr[\max(Y_i, 0) \ge a_i] \le \prod_i p^{a_i} = p^{\sum_i a_i} = p^{n}.$$
By a union bound, 
$\Pr[\sum_i Y_i  \ge 2n] \le \sum_{\vec{a} \in A} \Pr[\max(Y_i, 0) \ge a_i \text{ for all } i] \le |A| \cdot p^n.$
To complete the proof, it suffices to prove $|A| \le 4^n$. We can encode each $\vec{a} \in A$ as a binary string of $a_1$ zeros followed by a one, then $a_2$ zeros followed by a one, etc. As there are $\sum_i a_i = n$ zeros and $n$ ones, the string's length is $2n$, and $|A| \le 2^{2n} = 4^n$.
\end{proof}
Finally, combining the previous lemmas in the right way, we can extract the full bound:
\begin{theorem}[Chernoff Bound for Fair Coin Flips]
For $k \ge 1$, $\Pr[X \geq k \sqrt{n}] \le 2^{-\Omega(k^2)}$.
\label{thm:scratchpad}
\end{theorem}
\begin{proof}
Break the coins into $k^2$ groups of size $n/k^2 \pm 1$ each, and define $Y_1, Y_2, \ldots, Y_{k^2}$ so that $Y_i$ is the sum of the $X_i$s in group $i$. By Lemma \ref{lem:poorchernoff}, we have
$$\Pr\left[Y_i \ge j \sqrt{n / k^2}\right] \le 2^{-\Omega(j)}.$$
Thus there exists a positive constant $c$ such that $Y'_i := Y_i / (c \sqrt{n / k^2}) = Y_i / (c \sqrt{n} / k)$ satisfies $\Pr[Y'_i \ge j] \le 8^{-j}.$ The $Y'_i$s are independent geometric random variables, so we can apply Lemma \ref{lem:geo} (with $p = 1/8)$ to get
$$\Pr\left[\sum_{i = 1}^{k^2} Y'_i \ge 2k^2\right] \le 2^{-\Omega(k^2)}.$$
Plugging in $\sum X_i = \Theta(\sqrt{n} / k) \cdot \sum Y'_i$ proves the theorem.
\end{proof}

\section{Fair Coin Flips: The Same Proof But With Commentary}\label{sec:fair}

We will now repeat the proof in the previous section, but this time with ample additional commentary. The goal is to add flavor and intuition to the proof. Along the way, we will also prove a matching lower bound, concluding that $\Pr[X \ge k \sqrt{n}]$ is not just $2^{-\Omega(k^2)}$, but is actually $2^{-\Theta(k^2)}$. To simplify the exposition in this section, we will often ignore rounding errors when discussing division and square roots---alternatively, so that these rounding errors do not exist, you can feel free to imagine that we are focusing only on values of $n$ that are powers of four and $k$ that are powers of $2$. 

As before, let $X_1, X_2, \ldots, X_n$ be independent random coin flips, where $X_i = 1$ represents heads and $X_i = -1$ represents tails. Each $X_i$ independently satisfies $\Pr[X_i = -1] = \Pr[X_i = 1] = 0.5$. Let $X = \sum_i X_i$ count the total number of heads minus the total number of tails. We want to prove the following:

\begin{restatable}{theorem}{thmcoin}\emph{(Chernoff Bound for Fair Coin Flips) \phantom{f}}
For $k \in \{1, \ldots, \sqrt{n}\}$, 
\begin{equation}
\Pr[X \geq k \sqrt{n}] = 2^{-\Theta(k^2)}.
\label{eq:coin-flip-Chernoff}
\end{equation}
\label{thm:coin-flip-Chernoff}
\end{restatable}

We will present the proof of Theorem \ref{thm:coin-flip-Chernoff} in four bite-sized pieces. The first three pieces can be viewed as warm-up results, each of which has a very simple (almost straightforward) proof. Then, in the final piece, we will show how to combine the warm-up results in order to get the full theorem.

The first warm-up establishes what we call the \emph{Poor Man's Chernoff Bound}. This bound gets the wrong dependence on $k$, but it will be \emph{incredibly simple to prove}. Moreover (and perhaps surprisingly) the bound will play an important role in the proof of the full theorem.

\begin{restatable}{proposition}{proppoorman}\emph{(Poor Man's Chernoff Bound) \phantom{f}}
For even $k \le \sqrt{n}$, 
\begin{equation}
\Pr[X \geq k \sqrt{n}] \le 2^{- k/2}.
\label{eq:poor-man-Chernoff}
\end{equation}
\label{prop:poorman}
\end{restatable} 

The second warm-up establishes a very simple Chernoff bound for geometric random variables. This bound might seem like a niche special case, but we will see that it is actually a critical building block for getting tight Chernoff bounds (no matter what parameter regime you care about).
\begin{restatable}{proposition}{propgeom}\emph{(Sum of Geometric Random Variables) \phantom{f}}
Let $Y_1, Y_2, \ldots, Y_n$ be independent real-valued random variables and let $p \in (0, 1)$. Suppose each $Y_i$ satisfies for all non-negative integers $j$ that
\begin{equation}
    \Pr[Y_i \ge j] \le p^j.
    \label{eq:geom}
\end{equation}
Then the sum $Y = \sum_i Y_i$ satisfies
$$\Pr[Y \ge 2n] \le (4p)^{n}.$$
\label{prop:geometric}
\end{restatable}

The third warm-up result establishes the lower-bound side of Theorem \ref{thm:coin-flip-Chernoff}:
\begin{restatable}{proposition}{propcoinlower}\emph{(Fair Coins Lower Bound) \phantom{f}}
For $k \le \sqrt{n}$, 
\begin{equation}
\Pr[X \ge k \sqrt{n}] \ge 2^{-O(k^2)}.
\label{eq:lowercoins}
\end{equation}
\label{prop:coinlower}
\end{restatable}

Each warm-up individually has a very simple combinatorial proof. On the other hand, once we have completed the warm-ups, the full proof of Theorem \ref{thm:coin-flip-Chernoff} will be \emph{almost immediate}. This final part, where we put the pieces together to get the full theorem, is my favorite part of the proof. 


\subsection{The Poor Man's Chernoff Bound}

Our first warm-up is to prove Proposition \ref{prop:poorman}. 
\proppoorman*

We will make use of one basic fact:
\begin{equation}
\Pr[X \geq 2 \sqrt{n}] \le \frac{1}{4},
\label{eq:2stddev}
\end{equation}
which follows directly from Chebyshev's inequality. If you don't have Chebeyshev's inequality in cache, you can also feel free to take \eqref{eq:2stddev} as a black-box fact.

\begin{proof}[Proof of Poor Man's Chernoff Bound] 
Suppose we flip the $n$ coins one after another, so that $X_i$ gets revealed at time $i$. Say that we have achieved an \emph{upper deviation} of $R$ at time $t$ if $\sum_{i = 1}^t X_i = R$. We will be interested in the \emph{checkpoints} at which we first achieve upper deviations of $2 \sqrt{n}$, $4 \sqrt{n}$, $6 \sqrt{n}$, etc. That is, for $s = 1, 2, \ldots$, define the checkpoint $t_s$ to be the earliest point in time at which we have achieved an upper deviation of $2 s \sqrt{n}$. See Figure \ref{fig:poor_chernoff}.

\begin{figure}[h]
    \centering
    \includegraphics[scale = 0.5]{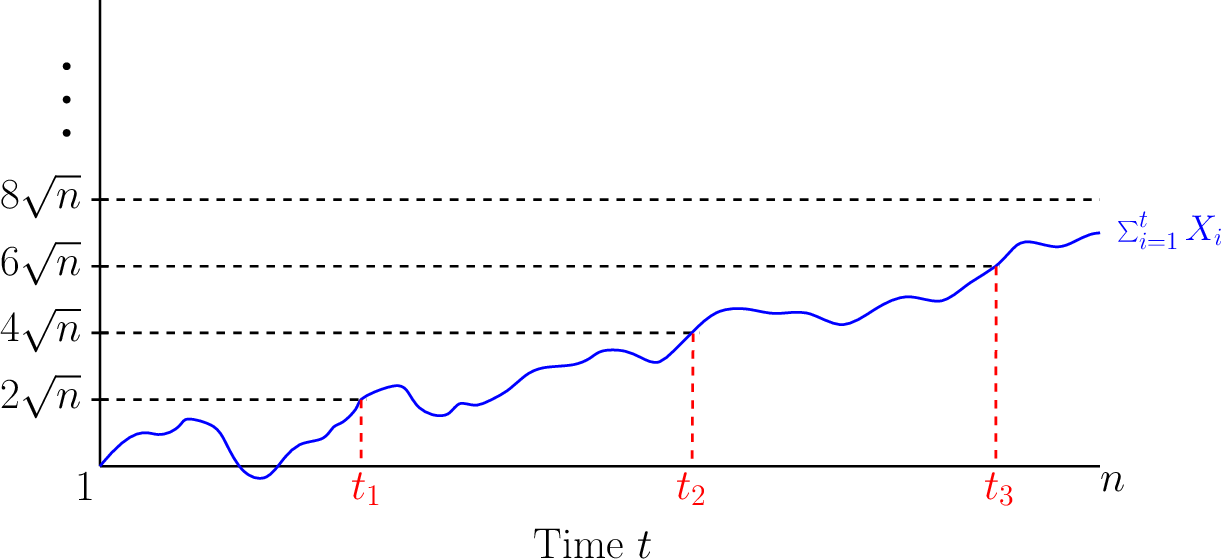}
    \caption{A graph of $\sum_{i = 1}^t X_i$ over time $t$, with labels for the times $t_1, t_2, t_3$ at which we first achieve upper deviations $2\sqrt{n}$, $4\sqrt{n}$, and $6\sqrt{n}$, respectively. The time $t_4$ does not exist in this example, because an upper deviation of $8\sqrt{n}$ is never achieved.}
    \label{fig:poor_chernoff}
\end{figure}

Notice that, \emph{a priori}, the checkpoint $t_1$ may not exist. (We may flip all $n$ coins and never get an upper deviation of $2 \sqrt{n}$). Even if $t_1$ exists, $t_2$ may not exist. And even if $t_2$ exists, $t_3$ may not, etc. Of course, if $X \ge k \sqrt{n}$ then the checkpoint $t_{k/2}$ \emph{must exist} (although the converse is not true). Thus, in order to bound $\Pr[X \geq k \sqrt{n}]$, we can instead bound $\Pr[t_{k / 2} \text{ exists}]$. 

Let's begin by proving that $\Pr[t_1 \text{ exists}] \le 1/2$. Observe that 
$$\Pr[X \ge 2 \sqrt{n}] = \Pr[t_1 \text{ exists}] \cdot \Pr[X_{t_1 + 1} + \cdots + X_n \ge 0 \mid t_1 \text{ exists}].$$
The probability on the left side is at most $1/4$ by \eqref{eq:2stddev}, and the second probability on the right side is at least $1/2$ by symmetry between heads/tails. Thus $1/4 \ge \Pr[t_1 \text{ exists}] \cdot 1/2$, implying that $\Pr[t_1 \text{ exists}] \le 1/2$. 

Next we argue that $\Pr[t_i \text{ exists} \mid t_1, \ldots, t_{i - 1} \text{ exist}] \le 1/2$ for any $i > 1$. Indeed, $t_i$ occurs only if, starting at time $t_{i - 1} + 1$, there is some point in time during the remaining $n - t_{i - 1} \le n$ coin flips at which we have again achieved an (additional) upper deviation of $2 \sqrt{n}$. However, we already know from our analysis of $\Pr[t_1 \text{ exists}]$ that any sequence of $\le n$ coin flips has probability at most $1/2$ of ever achieving upper deviation at least $2 \sqrt{n}$. Thus $\Pr[t_i \text{ exists} \mid t_1, \ldots, t_{i - 1} \text{ exist}] \le 1/2$. 

Putting the pieces together,
$$\Pr[X \geq k \sqrt{n}] \le \Pr[t_{k/2} \text{ exists}] \le \prod_{i = 1}^{k / 2} \Pr[t_i \text{ exists} \mid t_1, \ldots, t_{i - 1} \text{ exist}] \le \frac{1}{2^{k/2}}.$$
\end{proof}

It's worth taking a moment to understand the moral of the Poor Man's Chernoff bound. What the bound is really saying is that if we consider thresholds $0$, $2\sqrt{n}$, $4\sqrt{n}$, $6\sqrt{n}$, $\ldots$ for $X$, the marginal probability of getting from the $i$-th threshold to the $(i + 1)$-st is decreasing as a function of $i$. That is, the first upper deviation of $2\sqrt{n}$ is the easiest (occurring with probability roughly $1/2$). The next $2\sqrt{n}$ is the next easiest, and so on. This is an almost trivial fact (since each subsequent deviation has fewer coin flips to make use of than the previous ones), and, as we will see later on, it is also a fact that holds in many settings (not just coin flips). But even this simple fact is enough to get a nontrivial bound.

\subsection{A Simple Chernoff bound for Sums of Geometric Random Variables}

Our next warm-up is to prove Proposition \ref{prop:geometric}.

\propgeom*
\begin{proof}
    If $Y \ge 2n$, then $Y' = \sum_i \lfloor Y_i \rfloor$ must be at least $n$. Thus, there exists a tuple of non-negative integers $\langle q_1, q_2, \ldots, q_n\rangle$ such that $\sum_i q_i = n$ and such that $\max(Y_i, 0) \ge q_i$ for each $i$. Call such a tuple a \defn{witness sequence}. 
    
    We will complete the proof in two pieces. First, we bound the number of possible witness sequences by $4^n$. Next, we bound the probability of a given witness sequence occurring by $p^n$. Combining these facts, we have by a union bound that the probability of any witness sequence occurring is at most $4^n p^n \le (4p)^n$. 

    To bound the number of possible witness sequences, observe that each witness sequence  $\langle q_1, q_2, \ldots, q_n\rangle$ can be viewed as a way to throw $n$ balls into $n$ bins (i.e., place $q_i$ balls into each bin $i$). There is a classic trick for bounding the number of ways to do this: encode the witness sequence as a binary string with $n$ zeros and $n$ ones, where the string consists of $q_1$ ones, followed by a zero, then $q_2$ ones, followed by a zero, then $q_3$ ones, followed by a zero, and so on. This creates an injection from witness sequences to binary strings of length $2n$. Since there are trivially at most $2^{2n} = 4^n$ such binary strings, it follows that there are also at most $4^n$ possible witness sequences.

    To bound the probability of a given witness sequence occurring, we can simply apply \eqref{eq:geom}. This tells us that each $Y_i$ has probability at most $p^{q_i}$ of satisfying $Y_i \ge q_i$. As the $Y_i$s are independent, it follows that
    \begin{equation}\Pr[Y_i \ge q_i \text{ for all } i] \le \prod_{i = 1}^n \Pr[Y_i \ge q_i] \le \prod_{i = 1}^n p^{q_i} = p^n,
    \label{eq:Yichain}
    \end{equation}
    where the final equality makes use of the fact that $\sum_{i = 1}^n q_i = n$. This completes the proof.
\end{proof}

Note that, if the $Y_i$s are guaranteed to be integers, then the preceding argument gives us a slightly stronger bound (since we can use $Y$ in place of $Y'$). 
\begin{corollary}
    If the $Y_i$s are guaranteed to be integers, then 
    $$\Pr[Y \ge n] \le (4p)^{n}.$$
    \label{cor:geometric}
\end{corollary}

\subsection{The Lower Bound}

Our third warm-up is to prove the lower-bound side of our Chernoff bound. Although we don't typically prove the lower-bound side when we teach Chernoff bounds, we will see that its proof is \emph{remarkably} simple. 

\propcoinlower*

To prove \eqref{eq:lowercoins}, we will make use of another basic fact: 
\begin{equation}
\Pr[X \ge \sqrt{n}/4] \ge 1/4.
\label{eq:stddevlower}
\end{equation}
To streamline our exposition, we will take \eqref{eq:stddevlower} for granted. For completeness, however, we also include a simple combinatorial proof in Appendix \ref{app:lowercoins}.


\begin{proof}[Proof of Proposition \ref{prop:coinlower}]
Partition the coins into $k^2$ groups each of size $S = n / k^2$. Define $E_i$ to be the event that group $i$ achieves sum of at least $\sqrt{S}/4 = \sqrt{n / k^2}/4 = \sqrt{n} / (4k)$. Notice that, if \emph{all} of events $E_1, E_2, \ldots, E_{k^2}$ were to occur, then the total sum would be at least $k^2 \cdot \sqrt{n} / (4k) \ge k\sqrt{n}/4$. See Figure \ref{fig:lower}.

\begin{figure}[h]
    \centering
    \includegraphics[scale = 0.5]{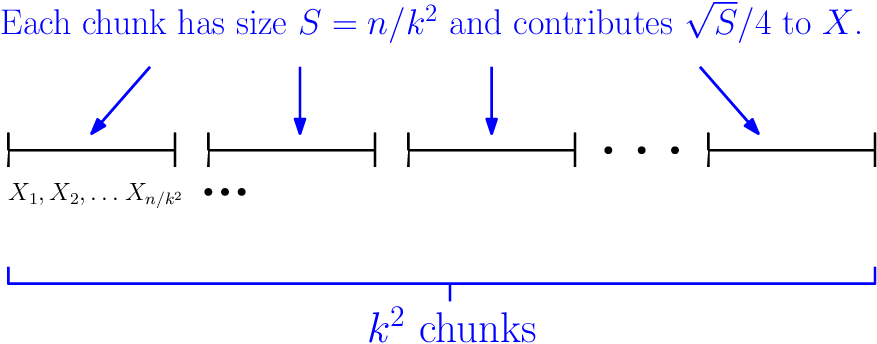}
    \caption{The lower-bound construction partitions the coins into $k^2$ groups, and considers the event that every group contributes $\Omega(\sqrt{S})$ to $X$, where $S$ is the size of each group. This would imply that $X \ge k^2 \cdot \Omega(\sqrt{S}) = \Omega(k^2 \cdot \sqrt{n / k^2}) = \Omega(k\sqrt{n})$.}
    \label{fig:lower}
\end{figure}

Applying \eqref{eq:stddevlower} to each group $i$ (with $n = |S|$), we see that each $E_i$ occurs with probability at least $1/4$. The probability that all of $E_1, E_2, \ldots, E_{k^2}$ occur is therefore at least $1/4^{k^2}$. Thus we have that
\begin{equation} \Pr[X \geq k \sqrt{n}/4] \ge 1/4^{k^2}.
\label{eq:0.1}
\end{equation}

This is not quite what we set out to prove, since we want $\Pr[X \ge k \sqrt{n}]$. Notice, however, that by a simple change of variables, \eqref{eq:0.1} implies that $\Pr[X \ge k \sqrt{n}] \ge (1/4)^{16 k^2}$ for all $k \le \sqrt{n}/4$. And the case of $k \ge \sqrt{n}/4$ follows from the simple fact that, for $k \in [\sqrt{n}/4, \sqrt{n}]$, we have $\Pr[X \ge k\sqrt{n}] \ge \Pr[X = n] = 2^{-n} = 2^{-O(k^2)}$. 
\end{proof}

It is important to understand why we broke the coins into $k^2$ groups, rather than some other number: $k^2$ is the magic number such that, if each group misbehaves just \emph{a little} (i.e., incurs a sum of at least $0.1\sqrt{S}$, which occurs with probability at least $0.1$), then the cumulative effect is a sum of $\Omega(k\sqrt{n})$. 

\subsection{The Upper Bound}

We are now prepared to prove the full Chernoff bound, restated below.

\thmcoin*

\begin{proof}
As we have already proven the lower bound (Proposition \ref{prop:coinlower}), we can focus here on the upper bound. We will show that
\begin{equation}
\Pr[X \geq 16 k\sqrt{n}] \le \frac{1}{4^{k^2}}.
\label{eq:upperbound}
\end{equation}

Our proof will build on each of the three warm-ups from earlier. We will use essentially the same group structure as in the lower bound, and we will analyze the groups by \emph{directly applying} Propositions \ref{prop:poorman} and \ref{prop:geometric}.

Break the coins into $k^2$ groups, each of size $S = n / k^2$. Define $C_1, C_2, \ldots, C_{k^2}$ to be the sums of the coin flips in each group. One way to think about the event $X \ge 16k \sqrt{n}$ is that the $C_i$s are, on average, at least $16 k \sqrt{n} / k^2 = 16 \sqrt{S}$. We know that for each group, having a sum of $16\sqrt{S}$ isn't very likely---in fact, by the Poor Man's Chernoff Bound (Proposition \ref{prop:poorman}), we know that each $C_i$ satisfies the bound
$$\Pr[C_i \ge 8 t\sqrt{S}] \le 2^{-4t}$$
for any positive integer $t$. If we define $\overline{C}_i = C_i / (8 \sqrt{S})$, then the Poor Man's bound translates to
$$\Pr[\overline{C}_i \ge t] \le 1/16^t.$$
In other words, $\overline{C}_i$ is bounded above by a geometric random variable.

We are interested in the event that
$$\sum_{i = 1}^{k^2} C_i \ge 16 k\sqrt{n}.$$
As noted above, this is equivalent to the event that, on average, each $C_i$ is at least $16\sqrt{S}$. Rewriting this in terms of the $\overline{C}_i$s, the event that we care about is
$$\sum_{i = 1}^{k^2} \overline{C}_i \ge 2k^2.$$
Since the $\overline{C}_i$s are independent geometric random variables, we can apply Proposition \ref{prop:geometric} to bound the probability of the above event by
$$(4/16)^{k^2} = 4^{-k^2},$$
which completes the proof.
\end{proof}

\section{Biased Coin Flips: The Large-Deviations Case}\label{sec:biased}

Next we extend our Chernoff bound to handle biased coin flips. Let $p \le 1/2$ be a probability. Suppose that each of $X_1, X_2, \ldots, X_n$ is $0$ with probability $1 - p$ and $1$ with probability $p$. As before, assume that the $X_i$s are independent, and set $X = \sum_i X_i$. Notice that we have swapped from each $X_i$ being in $\{-1, 1\}$ to each $X_i$ being in $\{0, 1\}$. This perspective, it turns out, will significantly simplify the exposition when we present the analysis for the large-deviation regime. 

Let $\mu = \E[X] = pn$. As discussed in the introduction, the Chernoff bound for $X$ splits into two parameter regimes. The \emph{small-deviation} regime is governed by a bound that looks very similar to what we had for fair coin flips: for $k \in \{1, 2, \ldots, \sqrt{\mu}\}$,
\begin{equation}
    \Pr[X \geq \mu + k \sqrt{\mu}] = 2^{-\Theta(k^2)}.
    \label{eq:smallskip}
\end{equation}
The \emph{large-deviation} regime is governed by a bound that looks a little different: for any $r \ge 2$ satisfying $1 \le r \mu \le n$, we have
\begin{equation}
    \Pr[X \geq r \mu] = 1 / \Theta(r)^{\Theta(r \mu)}. \label{eq:verylargedev}
\end{equation}
Note that \eqref{eq:verylargedev} takes a slightly different form than the version of the bound that we presented in the introduction, examining $\Pr[X \ge r \mu]$ rather than $\Pr[X \ge \mu + r \mu]$ -- this distinction, although only aesthetic (it changes the value of $r$ by $1$), will make our analysis a bit cleaner. 

The small-deviation case follows from almost exactly the same arguments as in the previous section, so we will skip its proof for now. (But, for completeness, it is worth noting that Theorems \ref{thm:general} and \ref{thm:generallower} in Section \ref{sec:general} directly imply \eqref{eq:smallskip}.) Instead, this section will focus on the large-deviation regime. What's neat is that the proof of \eqref{eq:verylargedev} will follow almost exactly the same structure as the proof that we have already seen.

\begin{proposition}
  Let $0 \le p \le 1/2$. Let $X_1, \ldots, X_n$ be i.i.d. $0$-$1$ random variables satisfying $\Pr[X_i = 1] = p$ and $\Pr[X_i = 0] = 1 - p$. Let $X = \sum_i X_i$ and let $\mu = \E[X] = pn$. For any $r \ge 2$ satisfying $1 \le r \mu \le n$, we have
$$\Pr[X \geq r \mu] = 1 / \Theta(r)^{\Theta(r \mu)}.$$
    \label{prop:largedeviation}
\end{proposition}

To simplify our discussion, we will assume in our proof of Proposition \ref{prop:largedeviation} that $n$ is divisible by $r \mu$. This is just to avoid some minor handling of rounding errors, and with a bit of casework one can actually show that this simplification is without loss of generality.

\begin{proof}
We begin with the lower bound. Break the coins into $r \mu$ groups. The coins in each group have cumulative expectation $p \cdot \frac{n}{r \mu} = p \cdot \frac{n}{r pn} = \frac{1}{r}$. With a bit of work, one can obtain the following basic fact: the probability of at least one coin in the group evaluating to $1$ is at least $\Omega(1/r)$.\footnote{Indeed, the probability that exactly one coin evaluates to $1$ is $\binom{n/(r\mu)}{1} p (1 - p)^{n / (r\mu) - 1}$. Since $n / (r\mu) = p^{-1} r^{-1}$, this probability is at least $r^{-1} (1 - p)^{p^{-1} - 1} \ge r^{-1} / e$.} It follows that, with probability at least $\Omega(1 / r)^{r \mu}$, every group will contribute at least $1$ to $X$, making for a total of at least $r \mu$. Thus
$$\Pr[X \geq r \mu] \ge \Omega(1/r)^{r \mu}.$$ 
This establishes the lower-bound direction.

To derive the upper bound, we need to first derive something that closely resembles the Poor Man's Chernoff Bound for the $X_i$s in a given group. Let $X_a, \ldots, X_b$ be the $X_i$s that comprise some group, and let $C = \sum_{i = a}^b X_i$. Since $\E[C]
 = 1/r$, we know from Markov's inequality that $\Pr[C \ge 1] \le 1/r$. That is, if we flip the coins $X_a, \ldots, X_b$ one after another, the probability that we ever get a $1$ is at most $1/r$. Similarly, once we get that 1, the probability that we ever get \emph{another} 1 in the same group is again at most $1/r$. Continuing like this, we can conclude that 
 \begin{equation}
     \Pr[C \geq k] \le r^{-k}.
     \label{eq:poormanlarge}
 \end{equation}
 
Now we can complete the proof using Proposition \ref{prop:geometric} (or, since the $X_i$s are integers, we can actually use Corollary \ref{cor:geometric}). Define $C_1, C_2, \ldots, C_{r\mu}$ so that $C_i$ is the sum of the $X_i$s in the $i$-th group. Equation \eqref{eq:poormanlarge} tells us that each $C_i$ is bounded above by a geometric random variable. It follows by Corollary \ref{cor:geometric} that
$$\Pr\left[\sum_{i = 1}^{r\mu} C_i \ge r\mu\right] \le (4/r)^{r\mu}.$$
Since $\sum_i X_i = \sum_i C_i$, this completes the proof of the upper bound. 

So, by following almost exactly the same template as before, we once again arrive at nearly matching upper and lower bounds.
\end{proof}

\section{Generalizing to Bennett's Inequality}\label{sec:general}

Part of what makes the proof approach in this paper useful is that the basic approach naturally extends to many other settings. To showcase, this, we will prove in this section an adaptive version of Bennett's inequality \cite{Deviations3}. More generally, the proof also extends to give Azuma's inequality \cite{Azuma67}, McDiarmid's inequality \cite{McDiarmid89}, and various Bernstein-type inequalities \cite{Bernstein1, Bernstein2} (and, indeed, in their most basic formulations, all of these inequalities are corollaries of Theorem \ref{thm:general}, hence our focus on Bennett's inequality). 



\begin{theorem}[Adaptive Version of Bennett's Inequality]
Let $n\in \mathbb{N}$ and $v \in \mathbb{R}^+$. Suppose that Alice selects $\mathcal{D}_1,\mathcal{D}_2,\ldots, \mathcal{D}_n $ where each $\mathcal{D}_i $ is a probability distribution over $ [-\infty, 1] $ with mean $0$ and with some variance $ v_i $. Alice selects $\mathcal{D}_1,\mathcal{D}_2,\ldots$ one at a time, and once a given $\mathcal{D}_i $ is selected, a random variable $X_i$ is drawn from the distribution $\mathcal{D}_i $. Alice gets to select the $\mathcal{D}_i $s (and thus also the $v_i$s) adaptively, basing $\mathcal{D}_i$ on the outcomes of $X_1, \ldots, X_{i - 1}$. The only constraint on Alice is that $\sum_{i = 1}^n v_i \le v$.

Define $X = \sum_{i = 1}^n X_i$. Then, for $k \in [1, \sqrt{v}]$, we have that
\begin{equation}\Pr[X \ge k\sqrt{v}] \le 2^{-\Omega(k^2)}. \tag{the small-deviation case}\end{equation}
And for $r \ge 2$, with $1 \le rv$, we have
\begin{equation} \Pr[X \ge r v] \le O(1/r)^{\Omega(rv)}. \tag{the large-deviation case}\end{equation}
\label{thm:general}
\end{theorem}

In the small-deviation case, and with a few extra constraints on Alice (namely that $X_i \in [-1, 1]$ and that $\sum v_i = v$), we can also get a matching lower bound. As far as I know, this lower bound has not appeared in past work, and does not follow from standard techniques.

\begin{theorem}[Lower Bound for Bennett's Inequality]
    Let $n, v \in \mathbb{N}$. Suppose that Alice selects $\mathcal{D}_1,\mathcal{D}_2,\ldots, \mathcal{D}_n $ where each $\mathcal{D}_i $ is a probability distribution over $ [-1, 1] $ with mean $0$ and with some variance $ v_i $. Alice selects $\mathcal{D}_1,\mathcal{D}_2,\ldots$ one at a time, and once a given $\mathcal{D}_i $ is selected, a random variable $X_i$ is drawn from the distribution $\mathcal{D}_i $. Alice gets to select the $\mathcal{D}_i $s (and thus also the $v_i$s) adaptively, basing $\mathcal{D}_i$ on the outcomes of $X_1, \ldots, X_{i - 1}$. The only constraint on Alice is that $\sum_{i = 1}^n v_i = v$.

Define $X = \sum_{i = 1}^n X_i$. Then, for $k \in [1, \sqrt{v}]$, we have that
$$\Pr[X \ge \Omega(k\sqrt{v})] \ge 2^{-O(k^2)}.$$
\label{thm:generallower}
\end{theorem}

It is worth noting that, as an immediate corollary of Theorem \ref{thm:general}, we also get a general-purpose Chernoff bound for non-iid real-valued coin flips (sometimes known as Hoeffding's bound).

\begin{corollary}[Chernoff Bound for Non-Identical Real-Valued Coin Flips \cite{Hoeffding94}]
Let $X_1, \ldots, X_n \in [0, 1]$ be independent random variables with means $p_1, \ldots, p_n$. Let $\mu = \sum_i p_i$ and let $X = \sum X_i$. Then, for any integer $k \le \sqrt{\mu}$,
\begin{equation}
    \Pr\left[X \ge \mu + k \sqrt{\mu}\right] \le 2^{-\Omega(k^2)}.
    \label{eq:hoeff}
\end{equation}
And for any $r \ge 2$ with $1 \le r\mu$,
$$\Pr\left[X \ge r \mu \right] \le O(1/r)^{\Omega(r\mu)}.$$
\label{cor:generalch}
\end{corollary}
\begin{proof}
    Each $X_i$ has variance $\E[X^2] - p_i^2 \le \E[X] \le p_i$. So the result follows from Theorem \ref{thm:general}.
\end{proof}

The rest of this section is structured as follows. Section \ref{sec:generalprelim} presents some basic prelimanaries, culminating in a variation of the Poor Man's Chernoff Bound that can be used in the proof Theorem \ref{thm:generallower}. Sections \ref{sec:gensmall} and \ref{sec:genlarge} then prove the small and large deviation cases, respectively, for Theorem \ref{thm:general}. Finally, Section \ref{sec:generallower} proves Theorem \ref{thm:generallower}.

\subsection{Preliminaries}\label{sec:generalprelim}

We begin by proving some preliminary lemmas that will be useful for both the upper and lower bounds. Our first lemma bounds the variance of $X$ by $v$.
\begin{lemma}
We have $\E[X^2] = \E[\sum_i v_i] \le  v$.
\label{lem:var}
\end{lemma}
\begin{proof}
Define $Y_i = \sum_{j = 1}^i X_j$. It suffices to show that for each $i \in [n]$, we have
$\E[Y_i^2 - Y_{i - 1}^2] = \E[v_i].$ By linearity of expectation,
$$\E[Y_i^2] = \E[(X_i + Y_{i - 1})^2] = \E[X_i^2] + \E[Y_{i - 1}^2] + 2\E[X_i Y_{i - 1}].$$
No matter the outcome of $Y_{i - 1}$, we have that $\E[X_i] = 0$, so $\E[X_i Y_{i - 1}] = 0$.
Thus 
$\E[Y_i^2] = \E[X_i^2] + \E[Y_{i - 1}^2] = \E[v_i] + \E[Y_{i - 1}^2],$
as desired.
\end{proof}

Let $Y_{\text{max}} = \max_j \sum_{i = 1}^j X_i$ be the largest sum achieved by any prefix of the $X_i$s.
Our next lemma uses Chebyshev's inequality to bound $Y_{\text{max}}$.
\begin{lemma}
For any $\ell \ge 1$, 
$\Pr[Y_{\text{max}}^2 \ge \ell v] \le \frac{1}{\ell}.$
\label{lem:chebyshev}
\end{lemma}
\begin{proof}
By Lemma \ref{lem:var}, we know that the variance of $X$ is at most $v$.  
By Chebyshev's inequality, it follows that
$\Pr[X^2 \ge \ell v] \le \frac{1}{\ell}.$

Notice, however, that if $Y_{\text{max}}^2 \ge \ell v$, then Alice can also force $X^2 \ge \ell v$:
as soon as $(\sum_{i = 1}^j X_i)^2 \ge \ell v$ for some $j$, she simply sets $X_{j + 1}, \ldots, X_{n}$ 
to be deterministically $0$. Thus any tail bound on $X$ implies the same tail bound on $Y_{\text{max}}$,
and the lemma is proven.
\end{proof}

Our next lemma establishes the key technical ingredient that we will need in order to obtain a Poor Man's Chernoff Bound. The lemma transforms bounds on $\Pr[Y_{\text{max}} \ge \alpha]$, for a given $\alpha$, into bounds on $\Pr[Y_{\text{max}} \ge \alpha + \beta + 1 \mid Y_{\text{max}} \ge \beta]$ for a given $\alpha, \beta$.
\begin{lemma}[Law of diminishing growth]
Let $q$ be a real number such that
$\Pr[Y_{\text{max}} \ge \alpha] \le q$
holds independently of Alice's strategy. For any $\beta > 0$, if Alice follows a strategy $\mathcal{A}$ such that $\Pr[Y_{\text{max}} \ge \beta] > 0$, then
$\Pr[Y_{\text{max}} \ge \beta + \alpha + 1 \mid Y_{\text{max}} \ge \beta] \le q.$
\label{lem:dimgrowth}
\end{lemma}
\begin{proof}
Suppose Alice follows strategy $\mathcal{A}$ and that $Y_{\text{max}} \ge \beta$. Let $j$ be the smallest $j$ such that
$\beta \le \sum_{i = 1}^j X_i \le \beta + 1.$
Define 
$Y'_{\text{max}} = \max_{k > j} \sum_{i = j + 1}^{k} X_i$ 
to be the maximum sum achieved by any prefix of $X_{j + 1}, X_{j + 2}, \ldots, X_n$. 

If $Y_{\text{max}} \ge \beta + \alpha + 1$, then we must have $Y'_{\text{max}} \ge \alpha$. On the other hand, by the definition of $q$, we have $\Pr[Y'_{\text{max}} \ge \alpha] \le q$.\footnote{Here, we are implicitly using the following observation: any strategy that Alice an use to make $Y'_{\text{max}}$ large could also be used to make $Y_{\text{max}}$ large, as Alice can set $X_1, X_2, \ldots, X_j := 0$ in order to force $Y'_{\text{max}} = Y_{\text{max}}$.}  Thus
$\Pr[Y_{\text{max}} \ge \beta + \alpha + 1 \mid Y_{\text{max}} \ge \beta] \le q.$
\end{proof}

Building on Lemma \ref{lem:dimgrowth}, we can immediately obtain a Poor Man's Chernoff Bound:
\begin{lemma}[Poor Man's Chernoff Bound]
If $v \ge 1$, then for any $k \in \mathbb{N}$,
\begin{equation}
\Pr[X \ge 4k \sqrt{v}] \le 1/4^{k}.
\label{eq:simple2}
\end{equation}
And if $v \le 1$, then for any $k \in \mathbb{N}$,
\begin{equation}
\Pr[X \ge k] \le v^{k / 2}.
\label{eq:simple1}
\end{equation}
\label{lem:poorgen}
\end{lemma}
\begin{proof}
If $v \ge 1$, then Lemma \ref{lem:chebyshev} bounds $\Pr[Y_{\text{max}}^2 \ge 4v] \le 1/4$ and thus $\Pr[Y_{\text{max}} \ge 2\sqrt{v}] \le 1/4$.
By Lemma \ref{lem:dimgrowth}, it follows that for any $j \in \mathbb{N}$,
we have $\Pr[Y_{\text{max}} \ge (2\sqrt{v} + 1) j] \le \frac{1}{4^j},$ and thus that $\Pr[Y_{\text{max}} \ge 4\sqrt{v} j] \le \frac{1}{4^j}.$

If $v \le 1$, then Lemma \ref{lem:chebyshev} bounds $\Pr[Y_{\text{max}}^2 \ge 1] \le v$ and thus $\Pr[Y_{\text{max}} \ge 1] \le v$.
By Lemma \ref{lem:dimgrowth}, it follows that for any integer $j \ge 0$,
$\Pr[Y_{\text{max}} \ge 1 + 2j] \le v^{1 + j},$
which implies \eqref{eq:simple1}.

\end{proof}

Of course, the bounds achieved in the previous lemma are much weaker than those stated by Theorem \ref{thm:general}. On the other hand, we achieved them with almost no work: we simply combined a trivial application of Chebyshev's inequality (Lemma \ref{lem:chebyshev}) with a natural observation about sums of random variables (Lemma \ref{lem:dimgrowth}). 

Finally, we will also need an extension of Proposition \ref{prop:geometric}, our Chernoff bound for Geomeric Random Variables:

\begin{restatable}{proposition}{propgeom2}
Let $Y_1, Y_2, \ldots, Y_n$ be real-valued random variables and let $p \in (0, 1)$. Suppose that, for each $i \in [n]$ and $j \in \mathbb{N}$, if we condition on any outcomes for $Y_1, Y_2, \ldots, Y_{i - 1}$, then $Y_i$ satisfies 
\begin{equation*}
    \Pr[Y_i \ge j \mid Y_1, \ldots, Y_{i - 1}] \le p^j.
\end{equation*}
Then the sum $Y = \sum_i Y_i$ satisfies
$$\Pr[Y \ge 2n] \le (4p)^{n}.$$
\label{prop:geometric2}
\end{restatable}
\begin{proof}
    The proof is the same as for Proposition \ref{prop:geometric}, except that the reasoning which we use to derive \eqref{eq:Yichain} now becomes:
        \begin{equation*}\Pr[Y_i \ge q_i \text{ for all } i] \le \prod_{i = 1}^n \Pr[Y_i \ge q_i \mid Y_1 \ge q_1, \ldots, Y_{i - 1} \ge q_{i - 1}] \le \prod_{i = 1}^n p^{q_i} = p^n.
    \end{equation*}
\end{proof}

With these preliminaries in place, we can now proceed to prove much stronger bounds using the same approach as in previous sections.

\subsection{Upper Bound for the Small-Deviation Regime}\label{sec:gensmall}

In this section, we establish an upper bound for the small-deviation regime: we prove that for any $k \le \sqrt{v}$, we have
\begin{equation}
\Pr[X \ge 33 k \sqrt{v}] \le \frac{1}{4^{k^2}}.
\label{eq:smallgeneral}
\end{equation}

Call an $ X_i $ \defn{oversized} if $v_i \ge v / k^2$. Since $\sum_i v_i  \le v$, there can be at most $k^2$ oversized $X_i$s. Since each $X_i$ satisfies $X_i \le 1$, the total contribution of oversized $X_i$s to $X$ is at most $k^2 \le k\sqrt{v}$. In the rest of the proof, we will assume without loss of generality that there are no oversized $X_i$s and that our task is to bound $\Pr[X \ge 32 k \sqrt{v}]$. 

Partition the random variables $X_1, X_2, \ldots, X_n$ into at most $k^2$ \defn{groups} such that the $X_i$s in each group have variances that sum to at most $2v / k^2$. We define the partition greedily: we end group $j$ and begin group $j + 1$ once the $X_i$s in group $j$ have sum of variances at least $v / k^2$. Note that, since the $v_i$s are random variables, even the outcome of which $X_i$s are in each group are random variables. 

Define $C_1, C_2, \ldots, C_{k^2}$ so that $C_i$ is the sum of the $X_j$s in the $i$-th group (or $0$ if no such group exists). Each $C_i$ has expected value $0$, variance at most $2v / k^2$ (by Lemma \ref{lem:var}), and standard deviation at most $\sqrt{2v / k^2} \le 2\sqrt{v} / k$. 

Throughout the rest of the proof, define $\sigma = 2 \sqrt{v} / k$ to be an upper bound on the standard deviation of each $C_i$.
The statement $X \ge 32 k \sqrt{v}$ is equivalent to the statement $X \ge 16 k^2 \sigma$.  Thus our goal is to bound
\begin{equation}
    \Pr[X \ge 16k^2 \sigma] \le \frac{1}{4^{k^2}}.
    \label{eq:sufficientgeneral}
\end{equation}
One should think of this as the probability that the average $C_i$ is at least $16$ standard deviations large.

Condition on any outcomes for $C_1, \ldots, C_{i - 1}$. Then, applying the Poor Man's Chernoff Bound (Lemma \ref{lem:poorgen}) to $C_i$, we have that 
$$\Pr[C_i \ge 4t \sigma \mid C_1, \ldots, C_{i - 1}] \le 1/4^t.$$
Rewriting this in terms of $\overline{C}_i := C_i / (8\sigma)$ gives
$$\Pr[\overline{C}_i \ge t \mid \overline{C}_1, \ldots, \overline{C}_{i - 1}] \le 1/16^t.$$
In other words, $\overline{C}_i \mid \overline{C}_1, \ldots, \overline{C}_{i - 1}$ is bounded above by a geometric random variable (with mean roughly $1/16$), and this holds no matter the outcomes of $C_1, \ldots, C_{i  -1}$.

It follows by Proposition \ref{prop:geometric2} that
$$\Pr\left[ \sum_{i = 1}^{k^2} \overline{C}_i \ge 2k^2 \right] \le 4^{-k^2}.$$
Rewriting this in terms of the $C_i$s gives \eqref{eq:sufficientgeneral}, as desired.

\subsection{Upper Bound for the Large-Deviation Regime}\label{sec:genlarge}

In this section, we establish an upper bound for the large-deviation regime: we prove that for any $r \ge 1$ and for any $v \ge 0$ such that $rv$ is a positive integer, 
\begin{equation}
\Pr[X \ge 3rv] \le \frac{1}{(32/r)^{rv/2}}.
\label{eq:largegeneral} 
\end{equation}
This implies the large-deviation case in Theorem \ref{thm:general}. 

Call an $ X_i $ \defn{oversized} if $ v_i \ge 1/r $. Since $\sum_i v_i  \le v$, there can be at most $rv$ oversized $X_i$s. Since each $X_i$ satisfies $X_i \le 1$, the total contribution of oversized $X_i$s to $X$ is at most $rv$. In the rest of the proof, we will assume without loss of generality that there are no oversized $X_i$s and that our task is to bound $\Pr[X \ge 2rv]$. 

Partition the random variables $X_1, X_2, \ldots, X_n$ into at most $rv$ \defn{groups} such that the $X_i$s in each group have variances that sum to at most $\frac{2}{r}$. (Note that, by Lemma \ref{lem:var}, applied just to the $X_i$s in the group, this implies that the sum of the $X_i$s in the group has variance at most $\frac{2}{r}$.) To construct the groups, we define the partition greedily, meaning that we end group $j$ and begin group $j + 1$ once group $j$ has sum of variances at least $\frac{1}{r}$. The final $X_i$ in the group has variance $v_i \le 1/r$ by assumption, so the sum of the variances in each group is at most $2/r$.

Define $C_1, C_2, \ldots, C_{rv}$ so that $C_i$ is the sum of the $X_j$s in the $i$-th group (or $0$ if no such group exists), conditioned on the outcomes of the previous groups $C_1, C_2, \ldots, C_{i - 1}$. Each $C_i$ has expected value $0$ and variance at most $2/r$ (by Lemma \ref{lem:var}). 

Condition on any outcomes for $C_1, \ldots, C_{i - 1}$. By our Poor Man's Chernoff Bound (Lemma \ref{lem:poorgen}) to $C_i$, we have that
$$\Pr[C_i \ge t \mid C_1, \ldots, C_{i - 1}] \le (2/r)^{t/2}.$$
In other words, $C_i \mid C_1, \ldots, C_{i - 1}$ is bounded above by a geometric random variable with mean $\Theta(\sqrt{1/r})$, and this holds no matter the outcomes of $C_1, \ldots, C_{i - 1}$. 

It follows by Proposition \ref{prop:geometric2} that
$$\Pr\left[\sum_{i = 1}^{rv} C_i \ge 2rv\right] \le (4\sqrt{2/r})^{rv} \le (32/r)^{rv/2}.$$
This completes the proof of \eqref{eq:largegeneral}.

\subsection{Lower Bound for the Small-Deviation Regime}\label{sec:generallower}

Finally, we prove Theorem \ref{thm:generallower}. That is, we wish to show that, if Alice is required to satisfy $\sum_i v_i = v$ (rather than just $\sum_i v_i \le v$), and if Alice is required to guarantee that each $X_i \in [-1, 1]$ (rather than $(-\infty, 1]$), then for any $k \le \sqrt{v}$, we have
\begin{equation}
\Pr[X \ge \Omega(k \sqrt{v})] \ge \frac{1}{2^{O(k^2)}}.
\label{eq:smallgenerallower}
\end{equation}

Partition the random variables $X_1, X_2, \ldots, X_n$ into $\Theta(k^2)$ \defn{groups} such that the variances of the $X_i$'s in each group sum to between $v / k^2$ and $2v / k^2$. We define the partition greedily, meaning that once a given group $j$ reaches $v / k^2$, and if the total amount of remaining variance is at least $v / k^2$, then we begin a new group $j + 1$.

Define $C_1, C_2, \ldots, C_{\Theta(k^2)}$ so that $C_i$ is the sum of the $X_j$s in the $i$-th group. Each $C_i$ has expected value $0$,  variance at least $v / k^2$ (by Lemma \ref{lem:var}), and standard deviation at least $\sqrt{v} / k$.

We will prove the following lemma:
\begin{lemma}
No matter the outcomes of $C_1, \ldots, C_{i - 1}$, we have
$$\Pr[C_i \ge \Omega(\sqrt{v} / k)] \ge \Omega(1).$$
\label{lem:1devgeneral}
\end{lemma}

It follows that, 
$$\Pr[\text{every } C_i \text{ satisfies }C_i \ge \Omega(\sqrt{v} / k)] \ge \Omega(1)^{O(k^2)} = \frac{1}{2^{\Omega(k^2)}}.$$
On the other hand, if every $C_i$ satisfies $C_i \ge \Omega(\sqrt{v} / k)$, then $X = \sum_{i = 1}^{\Theta(k^2)} C_i \ge \Omega(k \sqrt{v})$. 
Thus, if we can establish Lemma \ref{lem:1devgeneral}, then we will have also proven \eqref{eq:smallgenerallower}.

Before diving into the proof of Lemma \ref{lem:1devgeneral}, I should make a small apology. The following proof is somewhat more messy than I would like it to be. With that said, the result that we are proving is ancillary---it's only needed for the lower bound. In fact, the lower bound in this subsection is so rarely discussed that I have not been able to find \emph{any} examples of it being discussed or even mentioned in the literature.

\begin{proof}[Proof of Lemma \ref{lem:1devgeneral}]
We make use of two facts. The first is that $C_i$ has variance at least $v / k^2$, so 
\begin{equation}
\E[C_i^2] \ge v / k^2.
\label{eq:avarb}
\end{equation}
The second is that $|C_i|$ is unlikely to be large. For this, Lemma \ref{lem:poorgen} suffices, telling us that
$$\Pr[|C_i| \ge 4 j \sqrt{\operatorname{Var}(C_i)}] \le 1 / 4^j,$$
which, since $\operatorname{Var}(C_i) \le 2 v / k^2 \le (2 \sqrt{v} / k)^2$, means that
\begin{equation}
\Pr[|C_i| \ge 8 j \sqrt{v} / k] \le 1 / 4^j \le 1/2^j.
\label{eq:adevb}
\end{equation}

For notational convenience, define $Q = \frac{C_i}{\sqrt{v} / k}$. Our goal is thus to establish that $Q \ge \Omega(1)$ with probability $\Omega(1)$. Equations \eqref{eq:avarb} and \eqref{eq:adevb} translate to
\begin{equation}
\E[Q^2] \ge 1.
\label{eq:avar}
\end{equation}
and
\begin{equation}
\Pr[|Q| \ge 8 j] \le 1 / 2^j.
\label{eq:adev}
\end{equation}

It might seem strange that \eqref{eq:adev} would be useful in this proof, since \eqref{eq:adev} is an \emph{upper bound} on $|Q|$ and we want a \emph{lower bound} on $Q$. Importantly, however, \eqref{eq:adev} forces $\E[|Q|]$ and $\E[Q^2]$ to be almost completely determined by cases where $|Q|$ is small. Indeed, two immediate consequences of \eqref{eq:adev} are that, if $c$ is a sufficiently large positive constant, then
\begin{equation}
\E[|Q| \cdot \mathbb{I}_{|Q| \ge c}] \le 1/c
\label{eq:a1}
\end{equation}
and
\begin{equation}
\E[Q^2 \cdot \mathbb{I}_{|Q| \ge c}] \le 1/c.
\label{eq:a2}
\end{equation}

Take $c$ to be a sufficiently large positive constant and suppose for contradiction that $\Pr[|Q| \ge 1/c] \le 1/c^3$. Then by \eqref{eq:a2}, 
\begin{align*}
\E[Q^2] & \le \Pr[|Q| \le 1/c] \cdot 1/c^2 + \Pr[1/c \le |Q| \le c] \cdot c^2 + \E[Q^2 \cdot \mathbb{I}_{|Q| \ge c}] \\
        & \le 1 / c^2 + 1/c + 1/c < 1,
\end{align*}
which contradicts \eqref{eq:avar}. Thus $\Pr[|Q| \ge 1/c] \ge \Omega(1)$.

This doesn't complete the proof, because we are interested in $\Pr[Q \ge \Omega(1)]$, not $\Pr[|Q| \ge \Omega(1)]$. 
However, the fact that $\Pr[|Q| \ge 1/c] \ge \Omega(1)$ does establish that $\E[|Q|] \ge \Omega(1)$. Since $\E[Q] = 0$, 
it follows that if we define $Q' = Q \cdot \mathbb{I}_{Q \ge 0}$ to be the positive component of $Q$, then 
\begin{equation}
\E[Q'] \ge \Omega(1).
\label{eq:a3}
\end{equation}

Now take $c'$ to be a sufficiently large positive constant and suppose for contradiction that $\Pr[Q' \ge 1/c'] \le 1/c'^2$. Then by \eqref{eq:a1}, 
\begin{align*}
\E[Q'] & \le \Pr[Q' \le 1/c'] \cdot 1/c' + \Pr[1/c' \le Q \le c'] \cdot c' + \E[|Q| \cdot \mathbb{I}_{Q \ge c'}] \\
        & \le 1 / c' + 1/c' + 1/c',
\end{align*}
which, if $c'$ is sufficiently large, contradicts \eqref{eq:a3}. Thus there exists some positive constant $c'$ such that $\Pr[Q' \ge 1/c'] \ge 1/c'^2 = \Omega(1)$. This establishes that
$$\Pr[Q' \ge \Omega(1)] = \Omega(1),$$
which implies that $\Pr[Q \ge \Omega(1)] = \Omega(1)$ and thus that $\Pr[C_i \ge \Omega(\sqrt{v} / k)] \ge \Omega(1).$
\end{proof}

\section{Acknowledgments}

The author would like to thank Rose Silver and Thatchaphol Saranurak for their extensive feedback and comments on earlier versions of this paper. The author would also like to thank Gabe Schoenbach, Shyan Akmal, Nicholas Kocurek and an anonymous reviewer for catching several typos and bugs.

This work was partially supported by a Harvard Rabin Postdoctoral Fellowship and by a Harvard FODSI fellowship under NSF grant DMS-2023528. 

Parts of this research were completed while William was a PhD student at MIT, where he was funded by a Fannie and John Hertz Fellowship and an NSF GRFP Fellowship. William Kuszmaul was also partially sponsored by the United States Air Force Research Laboratory and the United States Air Force Artificial Intelligence Accelerator and was accomplished under Cooperative Agreement Number FA8750-19-2-1000. The views and conclusions contained in this document are those of the authors and should not be interpreted as representing the official policies, either expressed or implied, of the United States Air Force or the U.S. Government. The U.S. Government is authorized to reproduce and distribute reprints for Government purposes notwithstanding any copyright notation herein.

\bibliographystyle{plainurl} \bibliography{writeup}

\begin{thebibliography}{10}

\bibitem{Azuma67}
Kazuoki Azuma.
\newblock Weighted sums of certain dependent random variables.
\newblock {\em Tohoku Mathematical Journal, Second Series}, 19(3):357--367,
  1967.

\bibitem{Deviations3}
George Bennett.
\newblock Probability inequalities for the sum of independent random variables.
\newblock {\em Journal of the American Statistical Association},
  57(297):33--45, 1962.

\bibitem{Bernstein2}
Sergei Bernstein.
\newblock On a modification of {C}hebyshev’s inequality and of the error
  formula of {L}aplace.
\newblock {\em Ann. Sci. Inst. Sav. Ukraine, Sect. Math}, 1(4):38--49, 1924.

\bibitem{Bernstein1}
Sergei~N Bernstein.
\newblock On certain modifications of {C}hebyshev’s inequality.
\newblock {\em Doklady Akademii Nauk SSSR}, 17(6):275--277, 1937.

\bibitem{DeviationsBook}
St{\'e}phane Boucheron, G{\'a}bor Lugosi, and Pascal Massart.
\newblock {\em Concentration {I}nequalities: A {N}onasymptotic {T}heory of
  {I}ndependence}.
\newblock Oxford University Press, 2013.

\bibitem{chernoff1952measure}
Herman Chernoff.
\newblock A measure of asymptotic efficiency for tests of a hypothesis based on
  the sum of observations.
\newblock {\em The Annals of Mathematical Statistics}, pages 493--507, 1952.

\bibitem{Chernoff14}
Herman Chernoff.
\newblock A career in statistics.
\newblock {\em Past, Present, and Future of Statistical Science}, 29, 2014.

\bibitem{chung2006concentration}
Fan Chung and Linyuan Lu.
\newblock Concentration inequalities and martingale inequalities: a survey.
\newblock {\em Internet Mathematics}, 3(1):79--127, 2006.

\bibitem{chvatal1979tail}
Vasek Chv{\'a}tal.
\newblock The tail of the hypergeometric distribution.
\newblock {\em Discrete Mathematics}, 25(3):285--287, 1979.

\bibitem{Deviations2}
Victor~H de~la Pena, Michael~J Klass, and Tze~Leung Lai.
\newblock Self-normalized processes: exponential inequalities, moment bounds
  and iterated logarithm laws.
\newblock {\em Annals of Probability}, pages 1902--1933, 2004.

\bibitem{Deviations6}
Kacha Dzhaparidze and JH~Van~Zanten.
\newblock On {B}ernstein-type inequalities for martingales.
\newblock {\em Stochastic Processes and Their Applications}, 93(1):109--117,
  2001.

\bibitem{Deviations7}
Xiequan Fan, Ion Grama, Quansheng Liu, et~al.
\newblock Exponential inequalities for martingales with applications.
\newblock {\em Electronic Journal of Probability}, 20, 2015.

\bibitem{Deviations1}
David~A Freedman.
\newblock On tail probabilities for martingales.
\newblock {\em Annals of Probability}, pages 100--118, 1975.

\bibitem{Deviations4}
Erich Haeusler.
\newblock An exact rate of convergence in the functional central limit theorem
  for special martingale difference arrays.
\newblock {\em Zeitschrift f{\"u}r Wahrscheinlichkeitstheorie und Verwandte
  Gebiete}, 65(4):523--534, 1984.

\bibitem{Hoeffding94}
Wassily Hoeffding.
\newblock Probability inequalities for sums of bounded random variables.
\newblock In {\em The Collected Works of Wassily Hoeffding}, pages 409--426.
  Springer, 1994.

\bibitem{impagliazzo2010constructive}
Russell Impagliazzo and Valentine Kabanets.
\newblock Constructive proofs of concentration bounds.
\newblock In {\em International Workshop on Randomization and Approximation
  Techniques in Computer Science}, pages 617--631. Springer, 2010.

\bibitem{Deviations8}
Rasul~A Khan.
\newblock Lp-version of the dubins--savage inequality and some exponential
  inequalities.
\newblock {\em Journal of Theoretical Probability}, 22(2):348, 2009.

\bibitem{blog}
William Kuszmaul.
\newblock Chernoff bounds, part 2: Mechanizing the process.
\newblock
  \url{https://mathandmaking.wordpress.com/2016/11/16/chernoff-bounds-part-2-mechanizing-the-process/},
  2016.
\newblock Math and Making Blog Post.

\bibitem{Deviations13}
Emmanuel Lesigne and Dalibor Voln{\`y}.
\newblock Large deviations for martingales.
\newblock {\em Stochastic Processes and Their Applications}, 96(1):143--159,
  2001.

\bibitem{Deviations10}
Robert Liptser and Vladimir Spokoiny.
\newblock Deviation probability bound for martingales with applications to
  statistical estimation.
\newblock {\em Statistics \& probability letters}, 46(4):347--357, 2000.

\bibitem{Deviations14}
Quansheng Liu and Fr{\'e}d{\'e}rique Watbled.
\newblock Exponential inequalities for martingales and asymptotic properties of
  the free energy of directed polymers in a random environment.
\newblock {\em Stochastic processes and their applications},
  119(10):3101--3132, 2009.

\bibitem{McDiarmid89}
Colin McDiarmid.
\newblock On the method of bounded differences.
\newblock {\em Surveys in Combinatorics}, 141(1):148--188, 1989.

\bibitem{morin2017encoding}
Pat Morin, Wolfgang Mulzer, and Tommy Reddad.
\newblock Encoding arguments.
\newblock {\em ACM Computing Surveys (CSUR)}, 50(3):1--36, 2017.

\bibitem{mulzer2018five}
Wolfgang Mulzer.
\newblock Five proofs of chernoff's bound with applications.
\newblock {\em arXiv preprint arXiv:1801.03365}, 2018.

\bibitem{Deviations5}
Iosif Pinelis.
\newblock Optimum bounds for the distributions of martingales in {B}anach
  spaces.
\newblock {\em The Annals of Probability}, pages 1679--1706, 1994.

\bibitem{Deviations9}
Emmanuel Rio et~al.
\newblock Extensions of the {H}oeffding-{A}zuma inequalities.
\newblock {\em Electronic Communications in Probability}, 18, 2013.

\bibitem{Deviations11}
Emmanuel Rio et~al.
\newblock On {M}c{D}iarmid's concentration inequality.
\newblock {\em Electronic Communications in Probability}, 18, 2013.

\bibitem{steinke2017subgaussian}
Thomas Steinke and Jonathan Ullman.
\newblock Subgaussian tail bounds via stability arguments.
\newblock {\em arXiv preprint arXiv:1701.03493}, 2017.

\bibitem{Deviations12}
Sara~A van~de Geer.
\newblock On {H}oeffding’s inequality for dependent random variables.
\newblock In {\em Empirical Process Techniques for Dependent Data}, pages
  161--169. Springer, 2002.

\bibitem{wainwright2019high}
Martin~J Wainwright.
\newblock {\em High-dimensional statistics: A non-asymptotic viewpoint},
  volume~48.
\newblock Cambridge university press, 2019.

\end{thebibliography}

\appendix

\section{Proof of \eqref{eq:stddevlower}}\label{app:lowercoins}

In this appendix, we will give a simple combinatorial proof of the following basic lemma about coin flips:
\begin{lemma}
    Let $X_1, X_2, \ldots, X_n$ be fair and independent $\pm 1$ coin flips. Then,
    $\Pr[X \ge \sqrt{n}/4] \ge 1/4.$
    \label{lem:lowercoins}
\end{lemma}

Define $Z_i = \sum_{j = 1}^i X_i$. We can think of the $Z_i$s as following a random walk, starting at 0, and progressing $\pm 1$ with equal probability on each step. As a thought experiment, let us continue this random walk in perpetuity, so we extend $Z_1, Z_2, \ldots, Z_n$ to an infinite sequence $Z_1, Z_2, \ldots$. 

For each $r \ge 1$, let $t_r$ be the first time that the random walk reaches $\pm r$, that is, the smallest $t \ge 1$ such that $|Z_t| = r$. It is a simple exercise to show that $t_r$ exists with probability $1$.\footnote{Indeed, one can argue that every $r$ steps, there is a probability of at least $1/2^r$ that we escape the interval $(-r, r)$. It follows that, after $kr$ steps, the probability that we fail to escape the interval is $(1 - 1/2^r)^k$, which goes to $0$ as $k$ goes to infinity.}

The main step in proving Lemma \ref{lem:lowercoins} is to solve for $\E[t_r]$.

\begin{claim}
For every power of two $r = 2^i$, we have that
    $$\E[t_r] = r^2.$$
    \label{clm:tr}
    \end{claim} 
\begin{proof}
We can prove this by induction. Since $t_1 = 1$, it suffices to show that for all $r = 2^i > 1$,
    \begin{equation}\E[t_r] = 4\E[t_{r/2}].
    \label{eq:tr}
    \end{equation}
We can prove \eqref{eq:tr} with a simple thought experiment. Suppose we wish to solve for $\E[t_r]$. The expected time to get from $0$ to $\pm r/2$ is just $\E[t_{r/2}]$. Say, without loss of generality, that we get to $r/2$, rather than $-r/2$ first. From there, the expected time to get to either $0$ or $r$ is again $\E[t_{r/2}]$. At that point, we are either done (we have reached $r$), or we are back at $0$ (with fifty percent chance). In the latter case, we need to restart the entire process: our expected time to get to $\pm r$ is once again $\E[t_r]$. Thus, we have the following recursion:
$$\E[t_r] = 2 \E[t_{r/2}] + 0.5 \cdot \E[t_r].$$
This, in turn, implies \eqref{eq:tr}, which completes the proof.
\end{proof}

We can now prove Lemma \ref{lem:lowercoins} with a simple application of Markov's inequality.

\begin{proof}[Proof of Lemma \ref{lem:lowercoins}]   
   First observe that, if there exists any $i \le n$ for which $Z_i \ge \sqrt{n}/4$, then with probability at least $1/2$ we will also have $Z_n \ge \sqrt{n}/4$. This is because the portion of the random walk determined by $X_{i + 1}, \ldots, X_n$ has (by symmetry) at least a $1/2$ chance of being non-negative. It follows that, to prove $\Pr[Z_n \ge \sqrt{n}/4] \ge 1/4$, it suffices to show that
   $$\Pr[\exists  \; i \le n \text{ such that }Z_i \ge \sqrt{n}/4] \ge 1/2.$$
   
   Let $r$ be the power of two in the range $[\sqrt{n}/4, \sqrt{n}/2)$. If $t_r \le n$, then there exists $i \le n$ such that $Z_i \ge 0.1 \sqrt{n}$. It therefore suffices to show that
   $$\Pr[t_r \le n] \ge 1/2.$$
   By Claim \ref{clm:tr}, we have that $\E[t_r] = r^2 \le n/4$. It follows by Markov's Inequality that $\Pr[t_r \ge n] \le 1/4,$ and therefore that $\Pr[t_r \le n] \ge 3/4 \ge 1/2$, as desired.
\end{proof}

\end{document}